\documentclass[]{article}

\usepackage{moreverb}

\usepackage{graphicx}
\usepackage[
colorlinks,bookmarksopen,bookmarksnumbered,citecolor=red,urlcolor=red]{hyperref}

\usepackage{amssymb,amsmath,amsthm}

\usepackage{authblk}

\newtheorem{theorem}{Theorem}
\theoremstyle{plain}	\newtheorem{lemma}{Lemma}

\newcommand\BibTeX{{\rmfamily B\kern-.05em \textsc{i\kern-.025em b}\kern-.08em
T\kern-.1667em\lower.7ex\hbox{E}\kern-.125emX}}

\def\volumeyear{2017}

\newcommand{\R}{\mathbb{R}}

\newcommand{\E}{\mathbb{E}}
\newcommand{\bX}{\mathbf{X}}
\newcommand{\bY}{\mathbf{Y}}
\newcommand{\bT}{\mathbf{T}}
\newcommand{\given}{\: \vert \:}
\newcommand{\abs}[1]{\lvert#1\rvert}
\newcommand{\norm}[1]{\lVert#1\rVert}
\newcommand{\cind}{\perp}          

\newcommand{\BODM}{B}
\newcommand{\betao}{\beta'}
\newcommand{\ones}{\mathbf{1}}

\begin{document}


\title{Outcome Based Matching}

\author{Jonathan Bates$^{2,3,4}$ , Alexander Cloninger$^{1,4}$}


\maketitle

\begin{abstract}
We propose a method to reduce variance in treatment effect estimates in the setting of high-dimensional data.
In particular, we introduce an approach for learning a metric to be used in matching treatment and control groups.  The metric reduces variance in treatment effect estimates by weighting covariates related to the outcome and filtering out unrelated covariates.
\end{abstract}


\footnotetext[1]{Department of Mathematics, University of California, San Diego, 9500 Gillman Dr, La Jolla, CA 92093}

\footnotetext[2]{
	Manager, American Express, 200 Vesey St, New York NY 10285 
}
\footnotetext[3]{
	Lecturer, Yale School of Medicine, 333 Cedar Street, New Haven CT 06510
}
\footnotetext[4]{
	Center for Outcomes Research and Evaluation, Yale-New Haven Hospital, New Haven CT
}


\section{Introduction}

A reliable estimate of a treatment effect in the observational setting rests on careful correction for imbalance between treatment and control groups.
Researchers often apply matching to assemble a study sample with balanced treatment and control groups,
followed by estimation of the treatment effect for the balanced sample.
Matching has much intuitive appeal.  Ideally, matching yields pairs of similar individuals, one from the treatment group and one from the control group.  
In aggregate, the initial similarity of these pairs lends strong evidence to the 
effect of treatment if we observe significantly different outcomes.  This point of view suggests that ``similarity'' means \emph{similar in risk} with respect to the outcome we have in mind.

Common matching methods include propensity score matching (PSM) and Mahalanobis distance matching (MDM).
However, both exhibit some undesirable properties, 
and the need to develop a more robust form of matching method has been noted for some time \cite{gu1993comparison}.  For example,
recently King and Nielson \cite{kingnielson} showed that PSM can lead to highly variable 
treatment effect estimates as one prunes more units, called the ``propensity score paradox''.
Moreover, they demonstrate that PSM can at best reconstruct a randomized treatment assignment setting.

We propose to match units on a distance that considers the influence of each covariate on the outcome.  In particular, we propose that the similarity imposed on a given covariate should be proportional to its influence on the outcome.  This is in contrast to Mahalanobis distance, which imposes similarity on a given covariate according to its variance in the data.  This is particularly important in the Big Data setting (i.e.\ $p \gg 1$), where many or most covariates have no significant influence on the outcome, which we call \emph{noise} covariates.  By weighting covariates by their influence, these noise covariates may be eliminated.  However, without this weighting, which is the case with MDM, the influence of the noise covariates leads to poorly matched groups.  Our approach is supported by simulations in \cite{brookhart2006variable}, which show that outcome covariates should be included in the propensity model, 
that covariates that predict exposure, but not outcome, increase variance in the treatment effect estimate, and therefore, matching should be done on outcome-related covariates only.  Thus, we hypothesize that our matching distance produces a treatment effect estimator with less variance.  Furthermore, we demonstrate through simulations that weighting the influence of covariates with respect to outcome reduces error, specifically variance, in treatment effect estimates.  Lower variance means estimate of treatment effectiveness can be trusted locally.


\section{Background}

\subsection{Why Match? An Example of Confounding}

Consider a scenario of strong confounding with one binary risk factor $X$, one binary ``placebo'' treatment $T$, and one binary outcome $Y$.  These occur according to the following probabilities
\begin{align*}
P(X=1) &= 0.25 \\
P(T=1 \mid X=1) &= 0.95 \\
P(T=1 \mid X=0) &= 0.05 \\
P(Y=1 \mid X=1) &= 0.95 \\
P(Y=1 \mid X=0) &= 0.05 \,.
\end{align*}

The outcome does not depend on the treatment, that is, the true treatment effect is 0.  
However, confounding will lead to erroneously high treatment effect estimates for the regression estimator, but matching prior to estimating the treatment effect mitigates this bias.

Specifically, we generate $N=200$ units $\{(X_i, T_i, Y_i)\}_{i=1}^N$.
Then we fit $\E\{Y \given X=x, T=t\} = g(x,t) = \mathrm{logistic}(\beta_0 + \beta x + \gamma_0 t)$ using logistic regression.  
We estimate the treatment effect as $\frac{1}{N} \sum_{i=1}^N g(X_i,1) - g(X_i,0)$.
After 1000 runs, the 95\% bootstrap confidence interval for the treatment effect is $(0.09, 0.29)$, suggesting that the treatment has a positive effect on the outcome.   For comparison, we also estimate the treatment effect by using exact 1-1 matching\footnote{Each treated unit is uniquely matched to a control unit with the same risk factor.} followed by using the regression estimator.  This gives a confidence interval of $(-0.13, 0.04)$ for the treatment effect, which covers the ground truth treatment effect of 0.
\footnote{cf. IPython Notebook \texttt{\detokenize{example_treatment_effect_under_placebo_A.ipynb}}}


% ...
%

\subsection{The average treatment effect on the treated (ATT)}

Consider a sample with $N$ units and potential outcomes data 
$\{(X_i, T_i, Y_i(1), Y_i(0) )\}_{i=1}^N$, 
where $X_i$ is a vector of $p$ pretreatment covariates for unit $i$, $T_i \in \{0, 1\}$
is the treatment indicator, and $Y_i(1), Y_i(0)$ are the potential outcomes for unit $i$
under treatment and control, respectively.  Herein, we make the stable unit treatment value assumption (SUTVA) that the potential outcomes are fixed for each unit. 
The treatment effect for unit $i$ is defined to be
$\tau_i := Y_i(1) - Y_i(0)$, and the average treatment effect is given by 
$\E \{ Y_i(1) - Y_i(0) \}$. Our principle aim is to estimate the average treatment effect on the treated (ATT)
\begin{equation}
\tau := \E \{ Y_i(1) - Y_i(0) \given T_i = 1 \} \,.
\end{equation}
However, as we never observe $Y_i(0)$ for the treatment group, 
we cannot estimate the expectation directly. 

We often estimate $\tau$ by regression or by estimating the counterfactual outcome $Y_i(1-T_i)$ through matching, or by some combination of both, such as matching followed by regression. Standard matching approaches, including PSM and MDM, proceed by pairing treated and control units that are close in the space of pretreatment covariates with respect to a balancing score $b$.  For example, in practice $X_i$ is matched to $X_{m_i}$, where $T_{m_i} = 1-T_i$, if $b(X_i)$ and $b(X_{m_i})$ are sufficiently close.
Formally, a \emph{balancing score} $b$ is defined to be a function such that
$P(X \given b(X), T=0) = P(X \given b(X), T=1)$, equivalently, $X \cind T \given b(X)$ (cf.\ \cite{rr83}).
The finest balancing score is the identity map 
$\mathrm{id}_\mathcal{X}(X) := X$, and Rosenbaum and Rubin \cite{rr83} 
prove that the coarsest balancing score 
is the propensity score $\pi(X) := P(T=1 \given X)$.

That matching on a balancing score leads to a consistent estimate of the ATT requires several assumptions on the treatment assignment.
A treatment assignment is said to be \emph{regular} if it is individualistic, probabilistic, and unconfounded (cf.\,\cite{imbensrubinbook}).  \emph{Individualistic} assignment asserts that the probability of all treatments conditioned on all covariates and potential outcomes factors over the individual units as a fixed function of a unit's covariates and potential outcomes.
\emph{Probabilistic} assignment asserts that $0 < \pi(X_i) < 1$ for all $i$, 
where $\pi(X) := P(T=1 \given X)$ is called the \emph{propensity score}.
\emph{Unconfounded} assignment asserts that the potential outcomes and treatment assignment 
for a unit are conditionally independent given the unit's covariates, that is, $(Y(1),Y(0)) \cind T \given X$.
In this work, we assume that treatment assignment is regular.
See Appendix \ref{s:reg_assign} for a more formal review.

The following theorem justifies the process of matching on a balancing score
to estimate the ATT.
\begin{theorem}[Rosenbaum \& Rubin \cite{rr83}]
	Suppose treatment assignment is regular and $b(X)$ is a balancing score.
	Then
	\begin{equation}
	\tau = \E \{ \, \E \{ Y_i \given b(X_i), T_i=1 \}
	- \E \{ Y_i \given b(X_i), T_i=0 \} \given T_i = 1 \, \} 
	\end{equation}
	where $Y_i = T_i Y_i(1) + (1-T_i) Y_i(0)$ is the observed outcome.
\end{theorem}

\subsection{Matching to estimate the ATT}

Consider an arbitrary linear map $X \mapsto X A$ and 
associated distance\footnote{Or, a pseudo-metric if $A$ is not invertible.} 
\begin{equation}
d_A(X_i,X_j) := \norm{X_i A - X_j A}_2 \,.
\end{equation}
We use greedy matching without replacement\footnote{
	This algorithm sequentially defines the match $m_i$ for each treated unit $i$ to be the nearest control unit with resepect to $d_A$ that has not yet been matched to another treated unit (cf.\ Section 18.4 in \cite{imbensrubinbook}).} with $d_A$ to construct 
an injective map $m : I_t \to I_c$ from treated to control units.
\footnote{$I_t := \{i \mid T_i=1\}$ and $I_c := \{i \mid T_i=0\}$.}
Furthermore, given a caliper $\delta>0$, we define 
$I_t^{\delta} := \{ i \mid d_A(X_i,X_{m_i}) \leq \delta \}$.
We then estimate the ATT as
\begin{equation}
\hat{\tau}^{\mathrm{match}}_{A,\delta}
= \frac{1}{\abs{I_t^\delta}} \sum_{i \in I_t^\delta} Y_i  -  Y_{m_i} \,.
\end{equation}

For example, with $A=I$, the identity, and $\delta=0$, $\hat{\tau}^{\mathrm{match}}_{I,0}$ is the exact matching estimator.  Assuming regular treatment assignment, this estimator is a consistent estimator for $\tau$.  More generally, if $X \mapsto X A$ is a balancing score, then $\hat{\tau}^{\mathrm{match}}_{I,0}$ is a consistent estimator for $\tau$.
In practice, exact matching is rarely possible and so one must find close matches, that is, nonzero $\delta$, to estimate ATT.

%
%
%
%
When exact matching is not feasible, model-based approaches are often used to reduce the bias from the discrepancy $X_i - X_{m_i}$ between matches.
Specifically, after computing the matches $m : I_t \to I_c$, we regress on some specification of the outcome, for example, 
$Y_i = \beta_0 + X_i \beta + \gamma_0 T_i + T_i (X_i-\bar{X}) \gamma + \epsilon_i$, 
on the pooled sample of $2 \abs{I_t^\delta}$ units $I_t^\delta \cup m(I_t^\delta)$. 
The least squares coefficient for treatment from the regression provides an estimate of the ATT:
\begin{equation}
\hat{\tau}^{\mathrm{reg}} 
= \hat{\gamma}_0 \,.
\end{equation}
This specific approach to bias correction is called \emph{parallel regressions on covariates}.

\subsection{Limitations of PSM}

Propensity score matching (PSM) corresponds to matching with respect to the map 
$X \mapsto X \alpha$, assuming the specification $\pi(X) = \mathrm{logistic}(X \alpha)$.
Matching on the propensity score can at best reconstruct a randomized treatment assignment setting, and as such has been criticized for its inefficiency \cite{kingnielson}.

\subsection{Limitations of MDM}

Mahalanobis distance matching (MDM) corresponds to matching with respect to the map $X \mapsto X R$, where $R R^T = \hat{\Sigma}^{-1}$ is any Choleski decomposition and $\hat{\Sigma}$ is the sample covariance matrix for $X$.

MDM breaks down when erroneous measurements are present or in the case of a rare condition, both of which occur commonly in health record data.
Gu and Rosenbaum point these out \cite{gu1993comparison}. For example, a lab value may be miscoded with an extremely high value, leading to high variance for the covariate.  The high variance diminishes the influence of this covariate in the distance, although it may be highly relevant to the outcome of interest.  
On the other hand, a rarely coded event, say, history of falling, will evaluate to have a very low variance, and MDM will essentially force exact matching on this covariate, to the point of throwing the sample out if there are no like controls to be found.



%


\section{Outcomes Based Matching}

We propose first inferring the influence of each pretreatment covariate on the outcome
and then constructing a distance for matching that weights each covariate by its outcome-specific influence.
First, let $\bX = (\ones, \bX_1, \dotsc, \bX_p) \in \R^{N\times (p+1)}$, where $\bX_j$ 
is a column vector of the $j$th pretreatment covariates, and let $\mathbf{Y} \in \R^N$ denote the outcomes.
Let $\betao$ be the ordinary least squares estimate\footnote{
	More generally, 
	$\betao := \bX^\dagger \bY$, 
	where $\bX^\dagger = \lim_{\lambda \to 0} (\bX^T \bX + \lambda I)^{-1} \bX^T$
	is the Moore-Penrose pseudoinverse.}
$\arg \min_{\beta} \, \norm{\mathbf{Y} - \bX\beta}_2$.
With $\BODM := \mathrm{diag}(\abs{\betao_0}^{1/2}, \abs{\betao_1}^{1/2}, \dotsc, \abs{\betao_p}^{1/2})$, 
we consider matching with respect to the map
\begin{eqnarray*}
	X &\mapsto& X \BODM \\
	d_{\BODM}(X_i, X_j) &=& \norm{(X_i-X_j) \BODM}_2  .
\end{eqnarray*}
Note that if $\betao_i = 0$, then $d_{\BODM}$ is invariant to the $i$th pretreatment covariate.  

First, suppose the $X_i$ are iid and 
\begin{align}
\begin{split}
T_i &\sim \mathrm{Bernoulli}(\pi(X_i)) \\
Y_i &= X_i \beta+T_i \gamma_0 + \epsilon_i
\end{split}
\tag{$\ast$}
\label{eq:cond2}
\end{align}
where $\epsilon_i$ is iid random noise with mean 0.

%
%
%
\begin{theorem}\label{thm:regression}
	Under \eqref{eq:cond2},
	\begin{equation}
	\E \betao = \beta + \gamma_0 \,\E_{\bX} [ \bX^\dagger \pi(\bX) ]
	\end{equation}
\end{theorem}
%
%

\begin{proof}
	Note 
	\begin{equation}
	E [ \bX^\dagger \bT ] 
	= \E_{\bX} \E_{\bT \mid \bX} [ \bX^\dagger \bT ] 
	= \E_{\bX} [ \bX^\dagger \pi(\bX) ]
	\end{equation}
	and
	\begin{equation}
	\betao 
	= \bX^\dagger \bY 
	= \bX^\dagger (\bX \beta + \gamma_0 \bT + \epsilon)
	\end{equation}
\end{proof}

In order to better understand Theorem \ref{thm:regression}, we can state more in the context that $X$ is multivariate Gaussian with non-singular covariance matrix, and $\pi(X_i) = \alpha_0 + X_i\alpha$.
\begin{lemma}\label{lemma:normalRegression}
	Let $X\sim N(0,\Sigma)$ for non-singular $\Sigma$.  Then if $\alpha_i=0$ and $\beta_i=0$, the resulting $\E[\beta'_i]=0$.
\end{lemma}
The proof of Lemma \ref{lemma:normalRegression} centers around showing $\E[X^\dagger\pi(X)]$ is a multiple of $\alpha$ in the normal distribution setting, as the zero elements of $\beta$ are clear.  The proof is in Appendix \ref{appendix:proofLemma}.

\subsection{Reduction of variance by using outcomes}

Second, we show that variance in the estimate of the ATT decreases with OBM
relative to MDM as the number of covariates $p$ grows.
To see this, define a \emph{perfect matching} to be any matching
$(i,m(i))\in \mathcal{E}$ such that 
$\E \{Y_i(0) \given X_i \} = \E \{Y_{m(i)}(0) \given X_{m(i)} \}$.
In particular, let $\E \{Y_i(0) \given X_i=x \} 
= f(x_{j_1}, \dotsc, x_{j_{d}})$. Then exact matching
on covariates $\{ x_{j_1}, \dotsc, x_{j_{d}} \}$ yields a perfect matching.


Suppose we have
\begin{align}
\begin{split}
X_i \given T_i=0 &\sim U([0,1]^p) \\
X_i \given T_i=1 &\sim U([0,1]^p+\eta) \\
Y_i &= X_i \beta+T_i \gamma_0 + \epsilon_i
\end{split}
\tag{$\ast\ast$}
\label{eq:cond3}
\end{align}
where $\epsilon_i$ is iid random noise with mean 0.
\footnote{The notation $[0,1]^p+\eta$
	indicates the shifted hypercube 
	$[\eta_1,\eta_1+1] \times \dotsm \times [\eta_p, \eta_p+1]$.}

\begin{theorem}\label{thm:noiseFolding}
	Assume \eqref{eq:cond3}, $\mathrm{supp}(\beta) = K$,
	and $\lvert K \rvert = d < p$.  Let $m$ be a perfect matching.  
	Then the expected MDM distance is
	\begin{eqnarray*}
		\E\left[d_R(X_i, X_{m(i)})^2\right] 
		= \frac{1}{6} \sum_{j\not\in K} (\Sigma^{\dagger})_{j,j} + \eta_{K^c}^* (\Sigma^{\dagger}) \eta_{K^c},
	\end{eqnarray*}
	where $\eta_{K^c} = \begin{cases} \eta_i, & i\not\in K \\ 0, & i \in K\end{cases}.$
\end{theorem}

One way to interpret this theorem is that, 
when using Mahalanobis distance matching, 
every variable that doesn't affect outcomes 
adds variance into the distance between 
two perfectly matched points 
(i.e. people with the exact same features 
for any feature that affects outcome risk).  
On top of that, if there is a treatment propensity 
on variables that don't affect risk, 
this systematic bias only gets worse.  
That is because $\Sigma^\dagger$ is positive semi-definite, 
which means $\eta_{K^c}^* (\Sigma^{\dagger}) \eta_{K^c} \ge 0$.

\begin{theorem}\label{thm:noiseFoldingODM}
	Under the assumptions of Theorem \ref{thm:noiseFolding}, 
	the expected ODM distance is
	\begin{eqnarray*}
		\E\left[d_{\BODM}(X_i, X_{m(i)})^2\right] 
		&=& \sum_{j\in K^c\cap supp(\eta)} (\betao)_{j} \left(\frac{1}{6} + \eta_j^2\right).
	\end{eqnarray*}
\end{theorem}
The proof of Theorem \ref{thm:noiseFoldingODM} is virtually identical to the proof of Theorem \ref{thm:noiseFolding} with a different matrix multiplying the pairwise distances.
Note that the systematic 
variance 
of ODM distance is limited to the features that affect treatment propensity but don't affect risk.  The other features do not affect this distance.


\section{Simulation Study}

\subsection{The King-Nielson Simulation}

King and Nielson \cite{kingnielson} use simulation to evaluate model dependence,
in particular, showing that PSM leads to greater model dependence than MDM.  
We replicate their simulation to evaluate our method in this context.
In the simulation, 100 control units and 100 treated units were drawn uniformly
from the squares $[0,5]^2$ and $[1,6]^2$, respectively.  
The outcome was generated as $Y = X_1+X_2 + 2T+\epsilon$, $\epsilon \sim N(0,1)$.  After generating data, the regression estimator
given by parallel regressions on covariates is used to estimate ATT.  
The specification of the outcome model is assumed to be unknown, 
hence a collection of outcome models are fit.  
Specifically, 512 models\footnote{i.e. feature sets} are fit, 
corresponding to linear regression with up to 3rd order products of $X_1, X_2$.
The variance of the 512 estimates of the ATE, averaged over 100 runs of the simulation, is plotted against different levels of pruning for each of the matching methods. 

We repeat the King-Nielson simulation with outcome based matching.
(Figure \ref{fig:kns})

\subsection{Extension of the King-Nielson Simulation}

Furthermore, we generalize the original simulation and evaluate the methods under additional simulation scenarios.
In particular, we draw a control group uniformly on the hypercube $H = [0,5]^d$ and the treatment group uniformly on the shifted hypercube $H+\eta = [\eta_1,\eta_1+5] \times \dotsm \times [\eta_d, \eta_d+5]$.
The outcome depends linearly on the pretreatment covariates and treatment: $Y = \beta_0 + \beta X + T( \gamma_0 + \gamma X) + \epsilon$, $\epsilon \sim N(0,\sigma)$.

Note that we identify the ATT as the expectation of $Y_1-Y_0 \given T=1$ on the common support
of the treated and control units
$[\eta_1,5] \times \dotsm \times [\eta_d, 5]$.
For example, the ATT corresponding to $Y =X_1+X_2+ T (-1+X_1)+ \epsilon$, $p=2$, 
and $\eta=(1,1)$ is $\tau = \E_{X \sim U([1,5] \times [1, 5])} [-1+X_1]= 2$.

With $\eta=(1,1)$, the common support of the treated and control groups comprises approximately $64\%$ of the treated sample.  Hence, given ground truth knowledge of the sample distributions, we expect to prune $36\%$ of the original sample for the best estimate of ATT.

\begin{table}
	\caption{Summary of scenarios based on the King Nielson example.}
	\begin{center}
		\begin{tabular}{ |c|c|c|c|c| } 
			\hline
			scenario & $p$ & $\eta$ & outcome model & ATT \\
			\hline
			1 & $2$ & $\mathbf{1}_2$ & $Y =X_1+X_2+ 2T+ \epsilon$ & 2 \\  
			2 & $10$ & $(1,1,0,\dotsc,0)$ & $Y =X_1+X_2+ 2T+ \epsilon$ & 2 \\  
			3 & $2$ & $\mathbf{1}_2$ & $Y =2 X_1+ 0.2 X_2+ 2T+ \epsilon$ & 2 \\  
			4 & $2$ & $\mathbf{1}_2$ & $Y =X_1+X_2+ T (-1+X_1)+ \epsilon$ & 2  \\ 
			5 & $10$ & $(1,1,0,\dotsc,0)$ & $Y = X_1+X_2+ T (-1+X_1)+ \epsilon$ & 2   \\ 
			6 & $2$ & $\mathbf{1}_2$ & $Y = 2 X_1+ 0.2 X_2+ T (-1+X_1) + \epsilon$ & 2 \\ 
			7 & $5$ & $\mathbf{1}_5$ & $Y =X_1+X_2+ 2T+ \epsilon$ & 2  \\ 
			8 & $10$ & $\mathbf{1}_{10}$ & $Y =X_1+X_2+ 2T+ \epsilon$ & 2   \\ 
			9 & $15$ & $\mathbf{1}_{15}$ & $Y =X_1+X_2+ 2T+ \epsilon$ & 2 \\ 
			\hline
		\end{tabular}
	\end{center}
\end{table}

\begin{figure}
	\begin{tabular}{ccc}
		
		\includegraphics[width=0.3\textwidth]{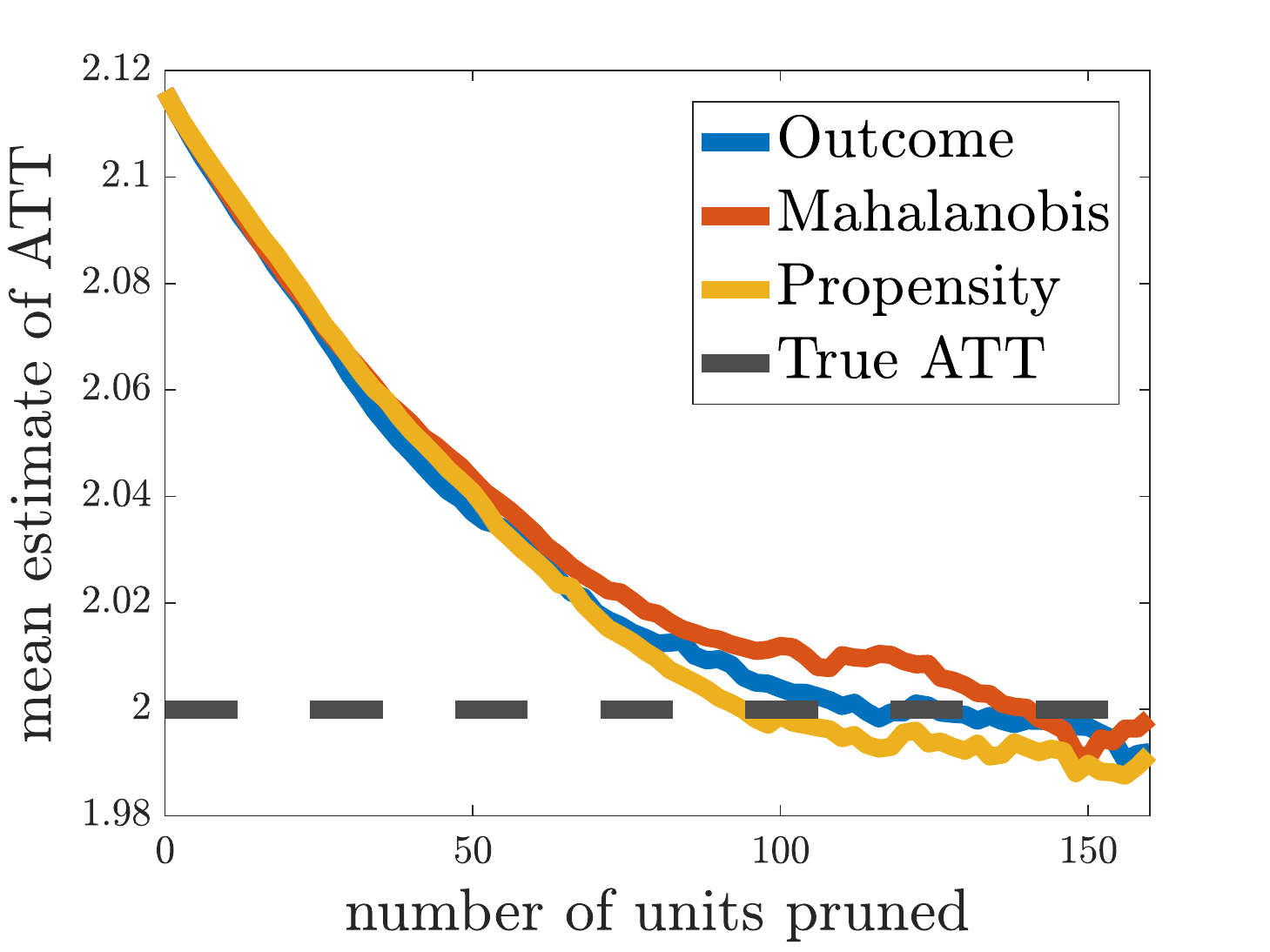}
		& 
		\includegraphics[width=0.3\textwidth]{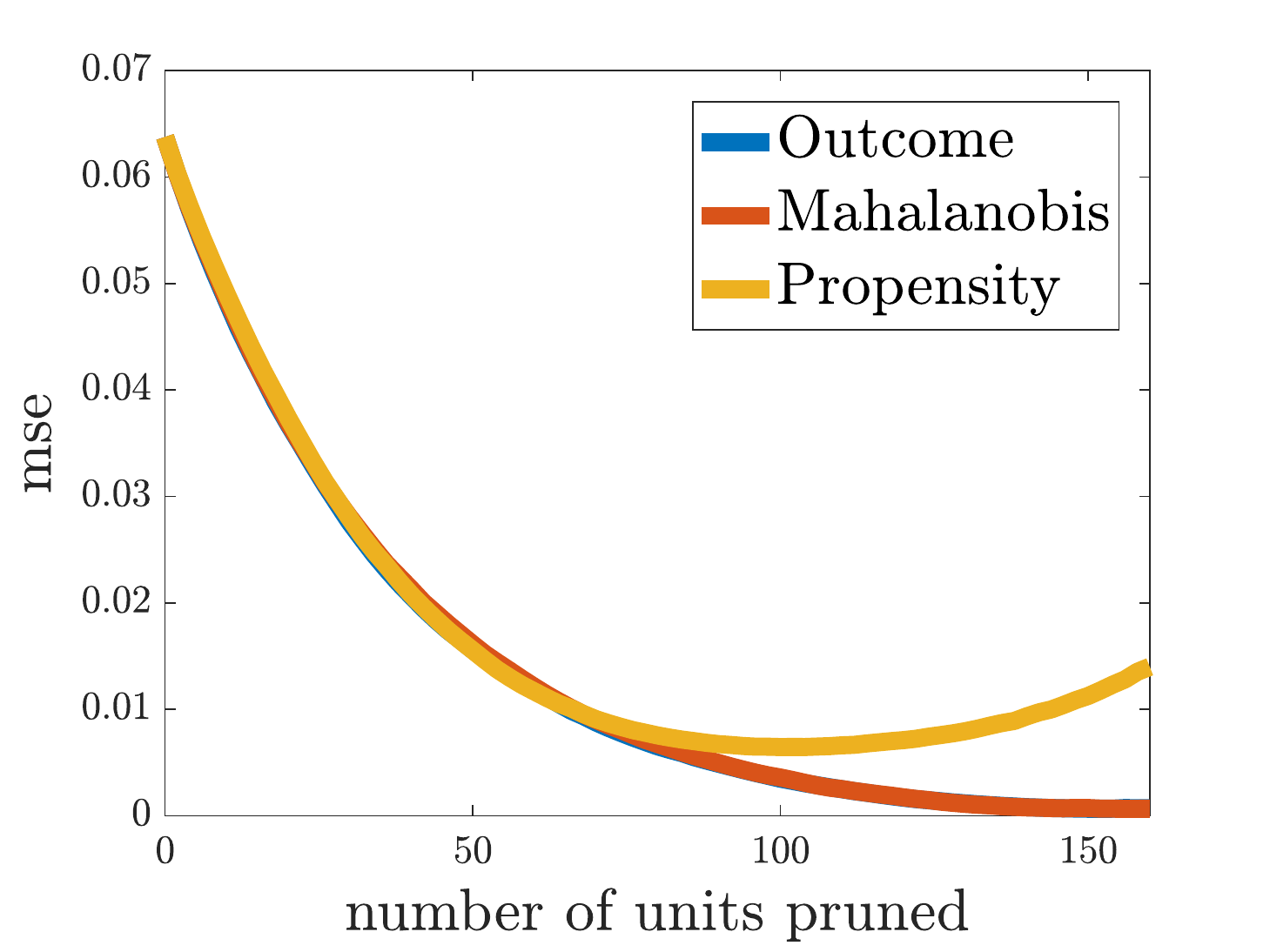} 
		&
		\includegraphics[width=0.3\textwidth]{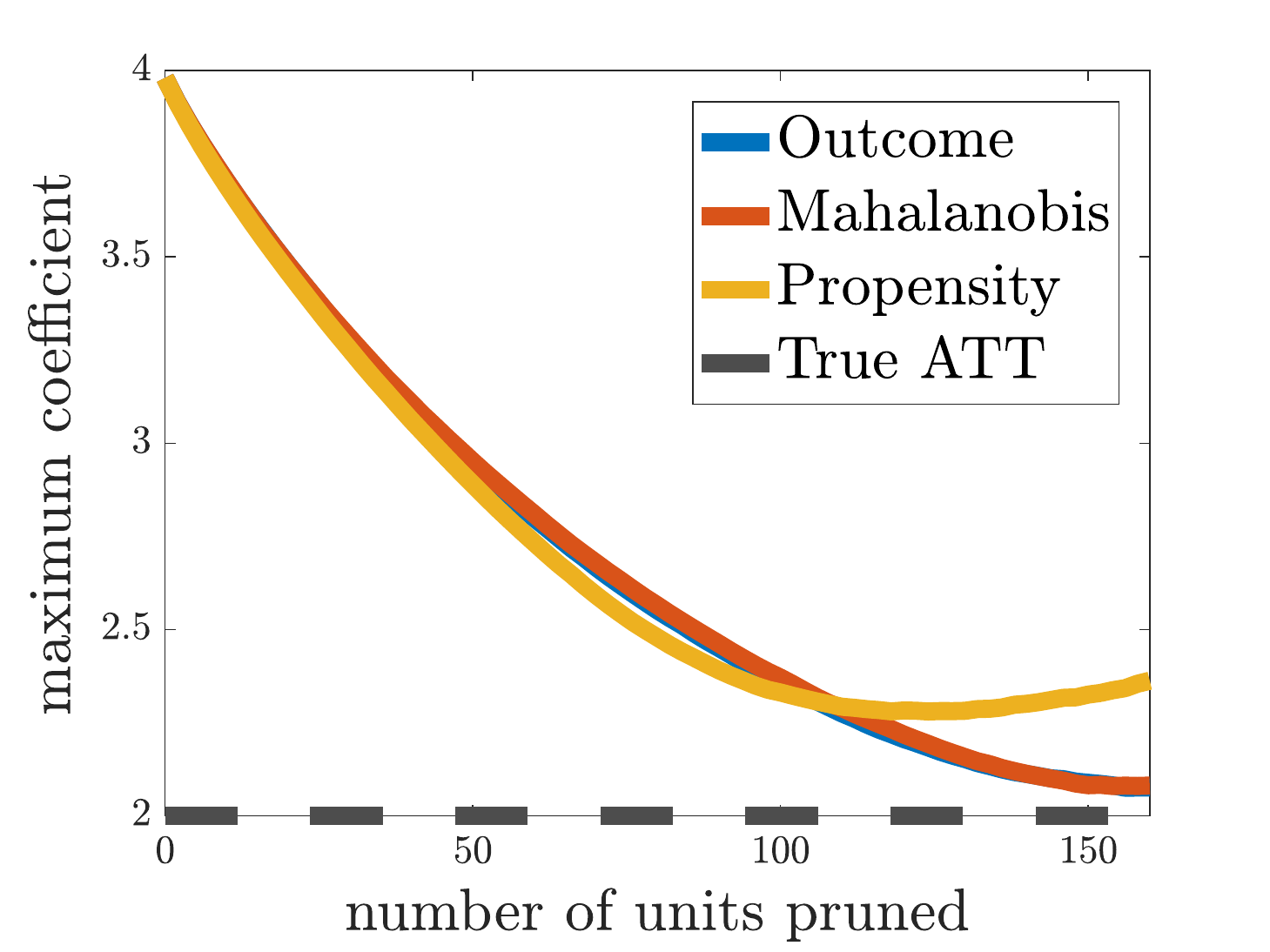}
		\\
		mean regression estimate ATT & MSE & max coefficient
	\end{tabular}
	\caption{King Nielson Simulation (Scenario 1) with OBM included.}
	\label{fig:kns}
\end{figure}

\begin{figure}
	\footnotesize
	\begin{tabular}{ccc}
			\includegraphics[width=0.3\linewidth]{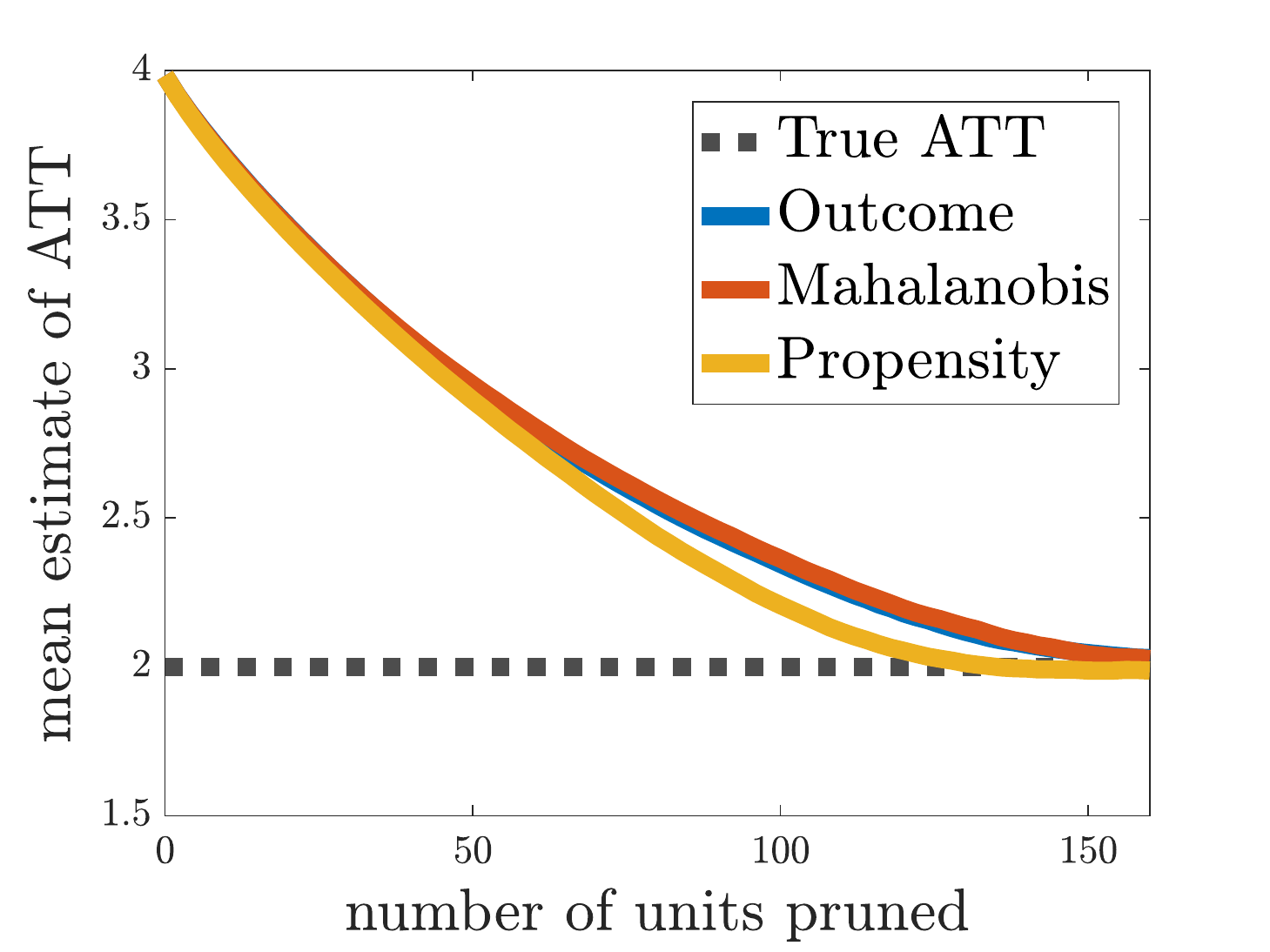} &
			\includegraphics[width=0.3\linewidth]{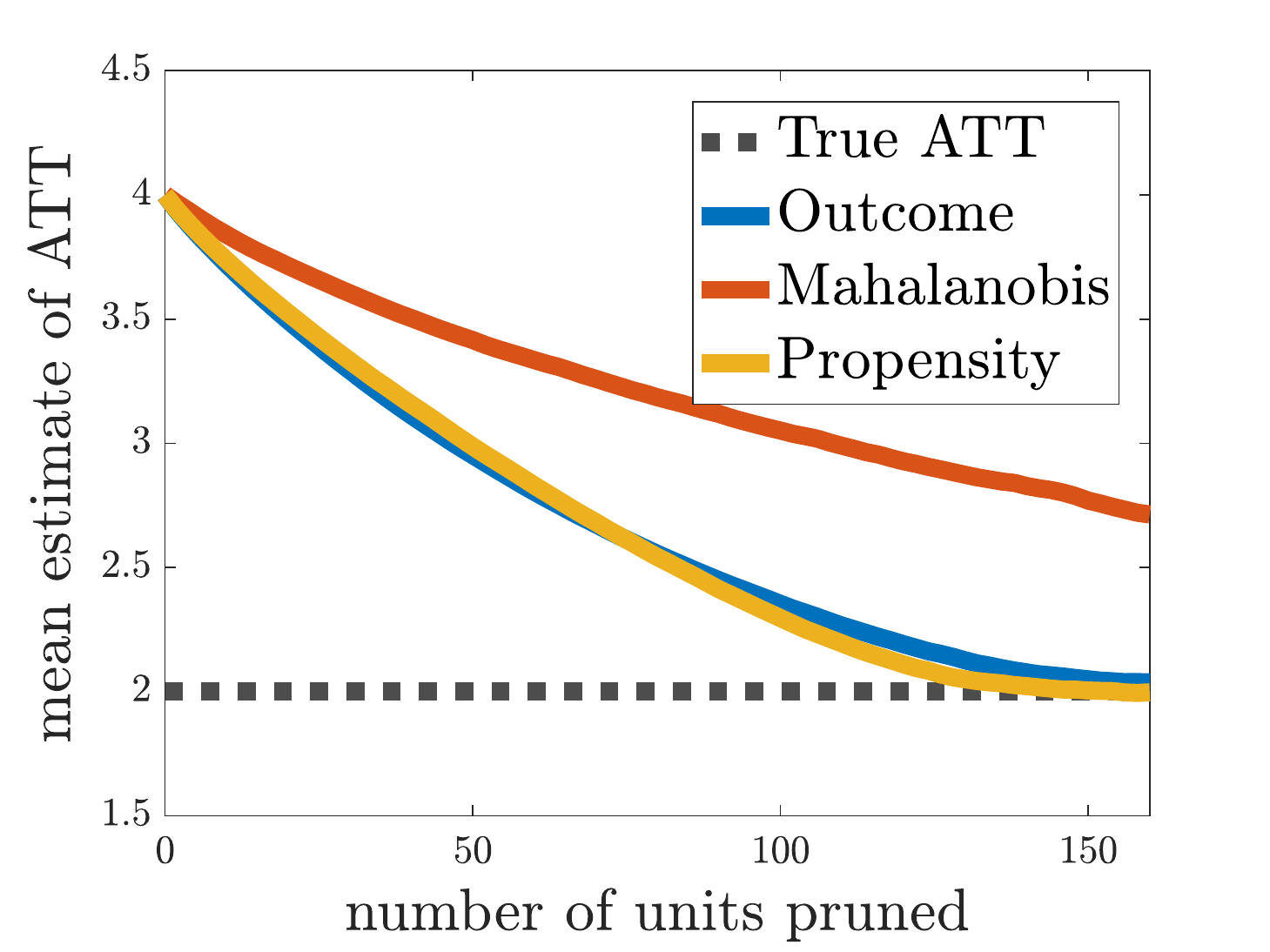} &
			\includegraphics[width=0.3\linewidth]{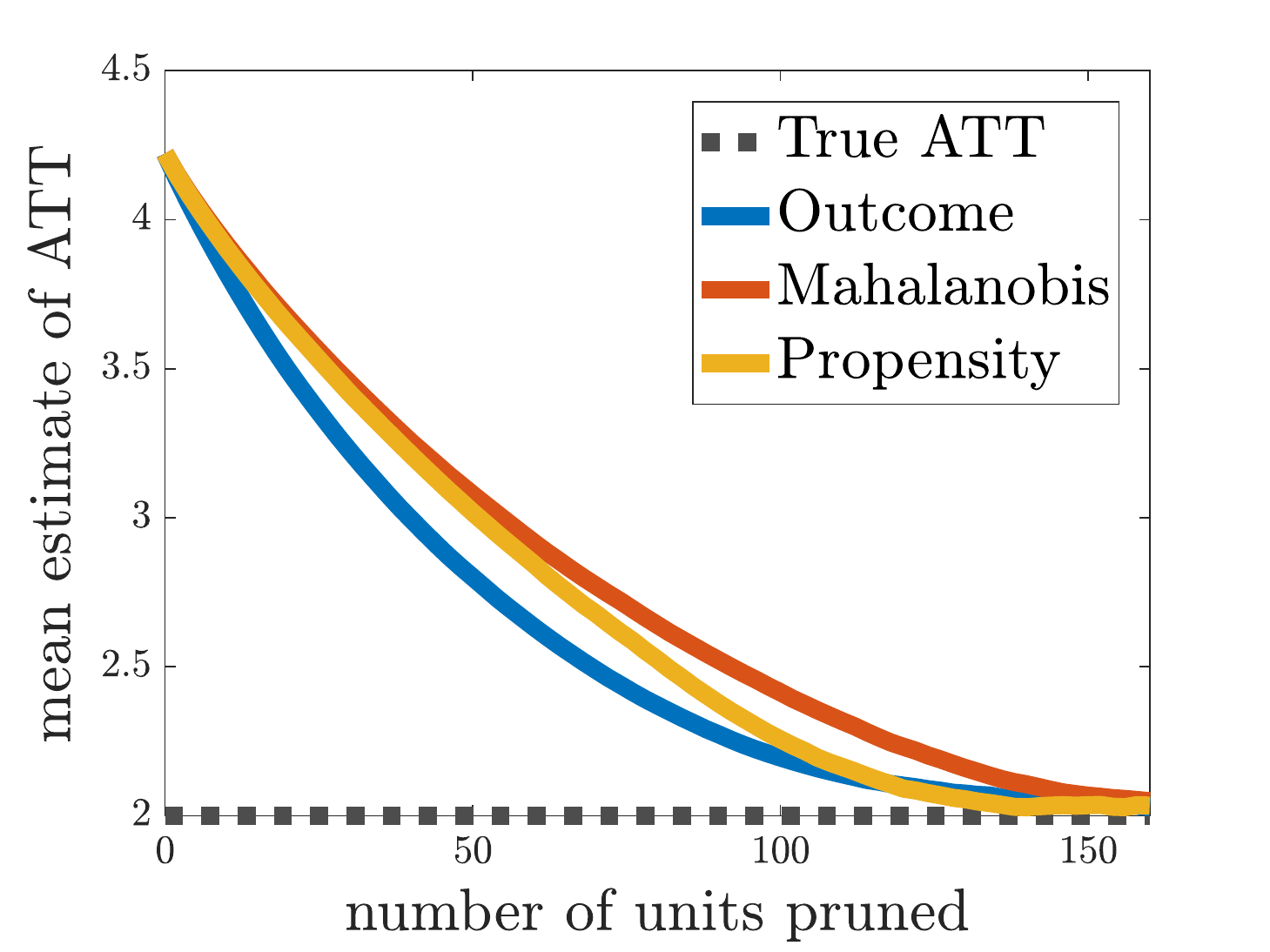} 
		\\
		Scenario 1 & Scenario 2 & Scenario 3 
		\vspace{2em} \\
			\includegraphics[width=0.3\linewidth]{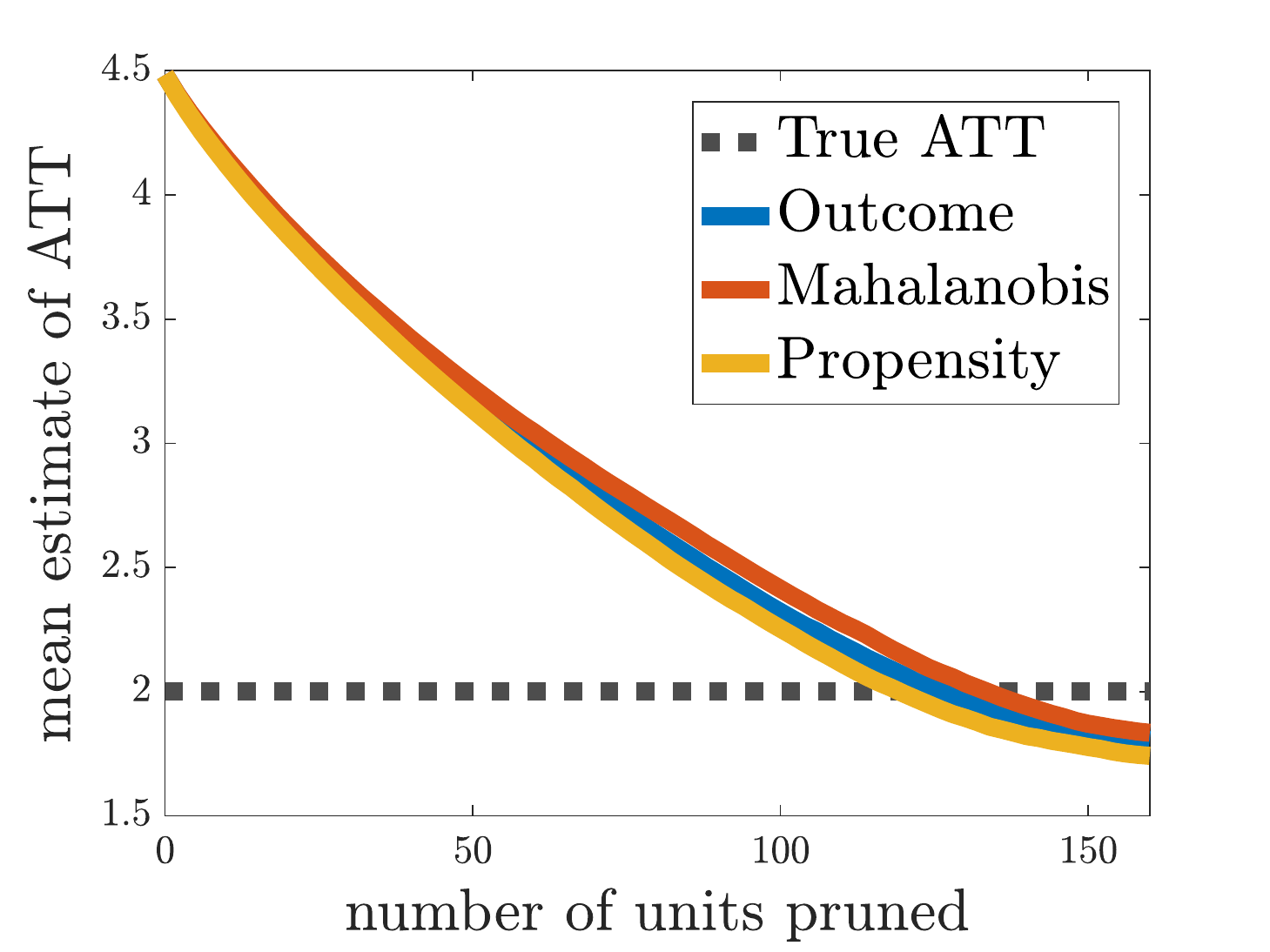} &
			\includegraphics[width=0.3\linewidth]{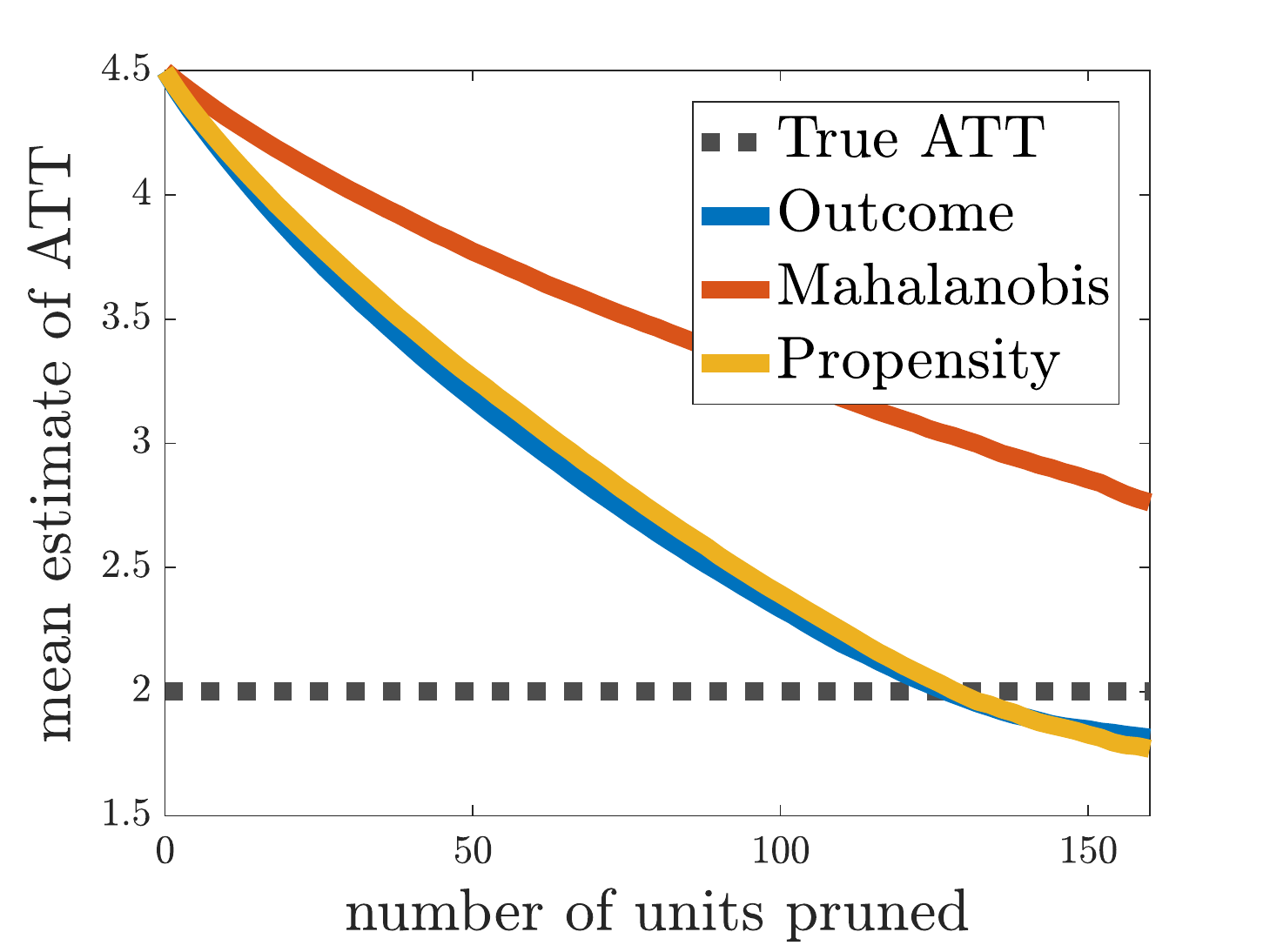} &
			\includegraphics[width=0.3\linewidth]{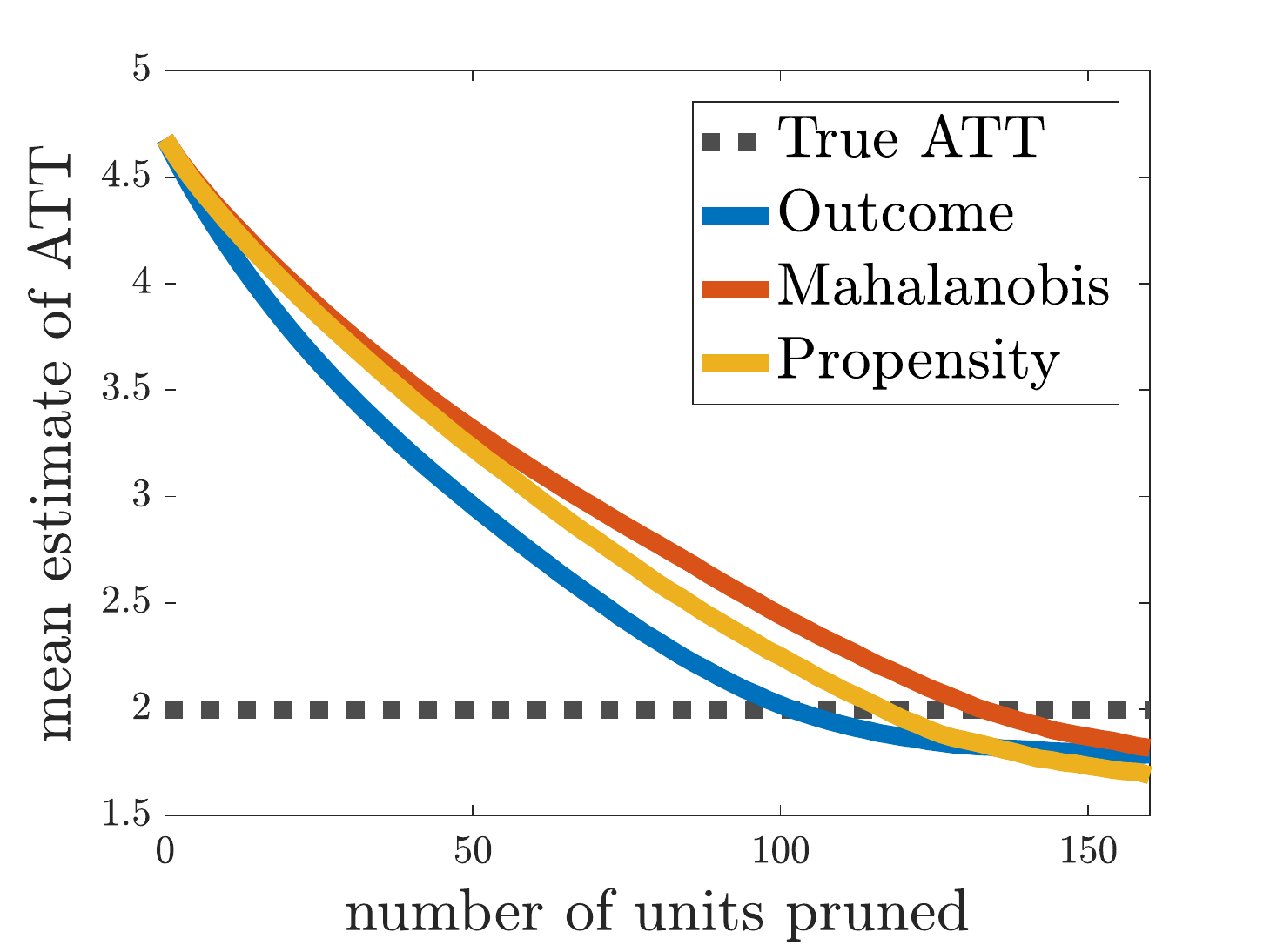} 
		\\
		Scenario 4 & Scenario 5 & Scenario 6
	\end{tabular}
	\caption{Mean matching estimates of ATT over the number 
		of units pruned for scenarios 1-6.}
	\label{fig:sc_123456_mse}
\end{figure}
\normalsize



\begin{figure}
	\footnotesize
	\begin{tabular}{ccc}
			\includegraphics[width=0.3\linewidth]{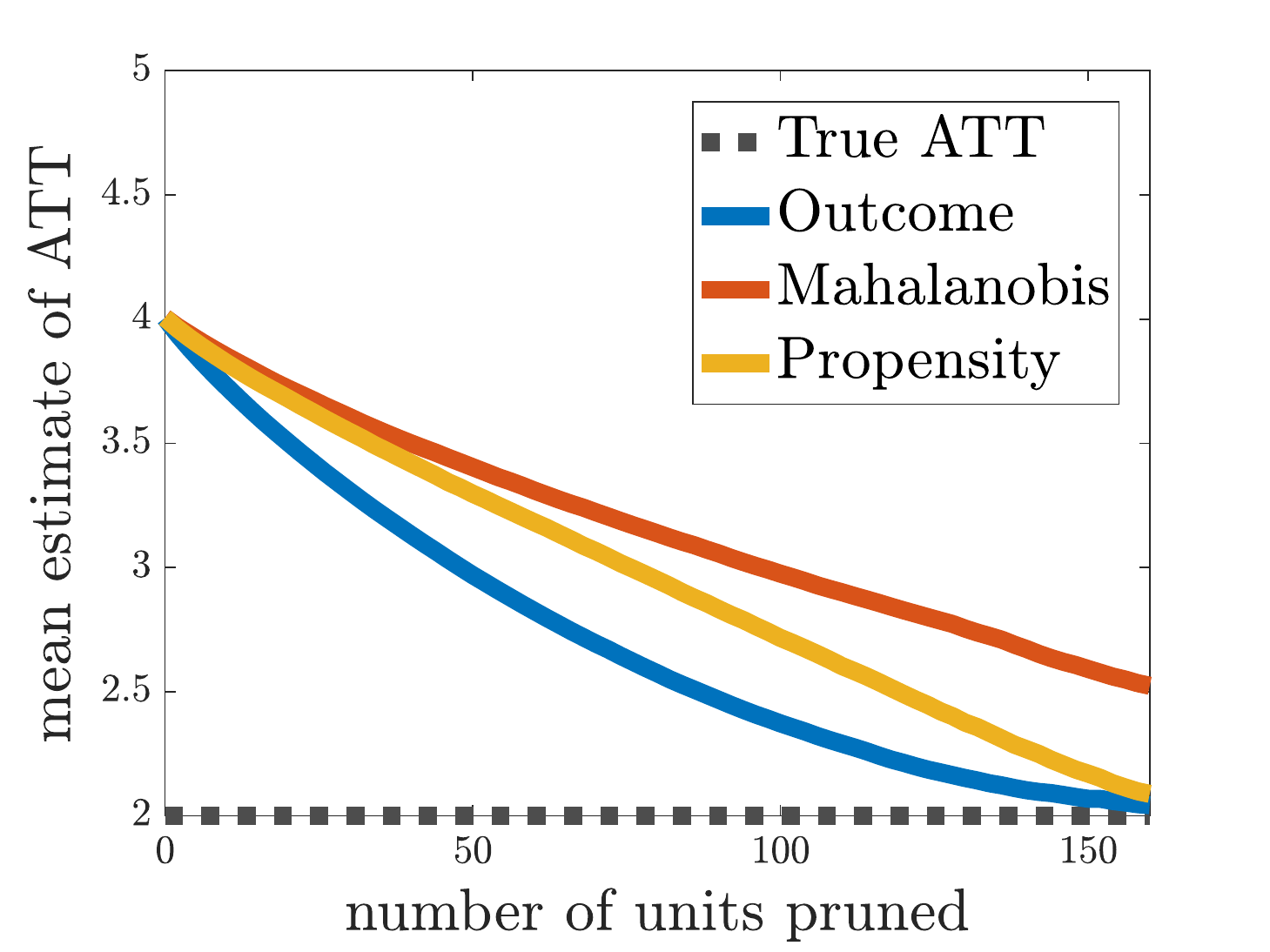} &
			\includegraphics[width=0.3\linewidth]{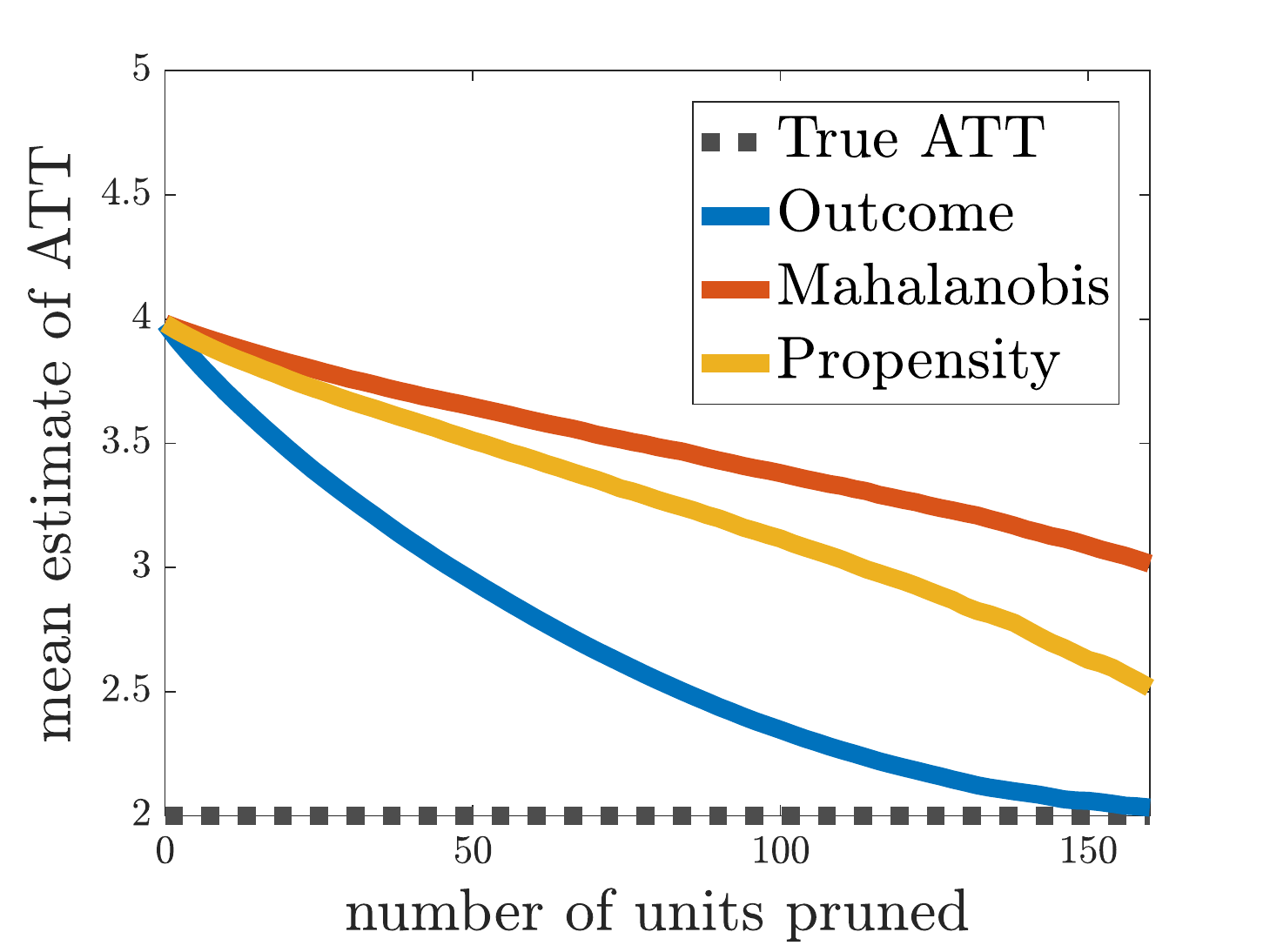} &
			\includegraphics[width=0.3\linewidth]{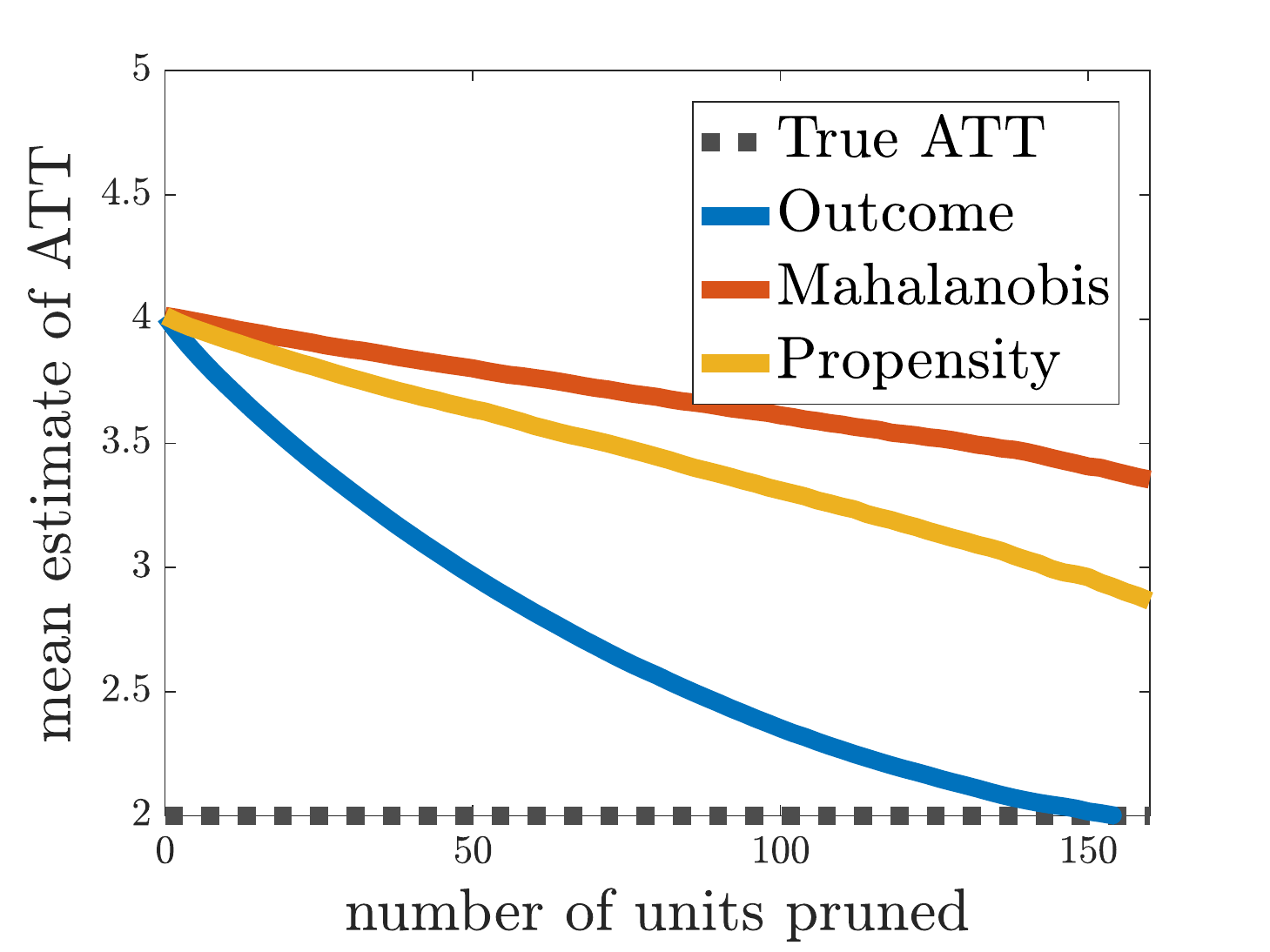} 
		\\
		Scenario 7 & Scenario 8 & Scenario 9
	\end{tabular}
	\caption{Mean matching estimates of ATT over the number 
		of units pruned for scenarios 7-9.}
	\label{fig:sc_789_mse}
\end{figure}
\normalsize

\begin{figure}
	\footnotesize
	\begin{tabular}{ccc}
			\includegraphics[width=0.3\linewidth]{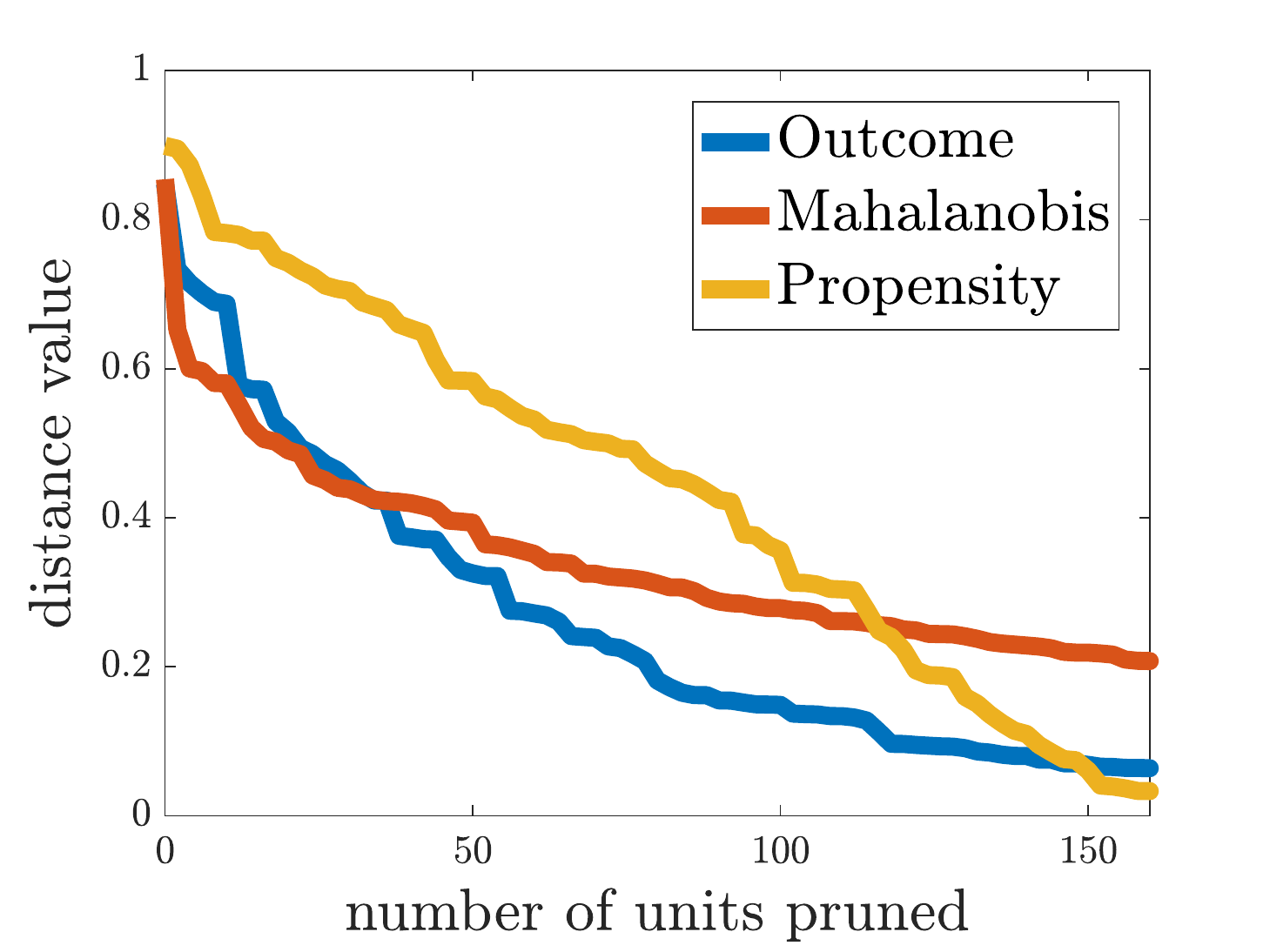} &
			\includegraphics[width=0.3\linewidth]{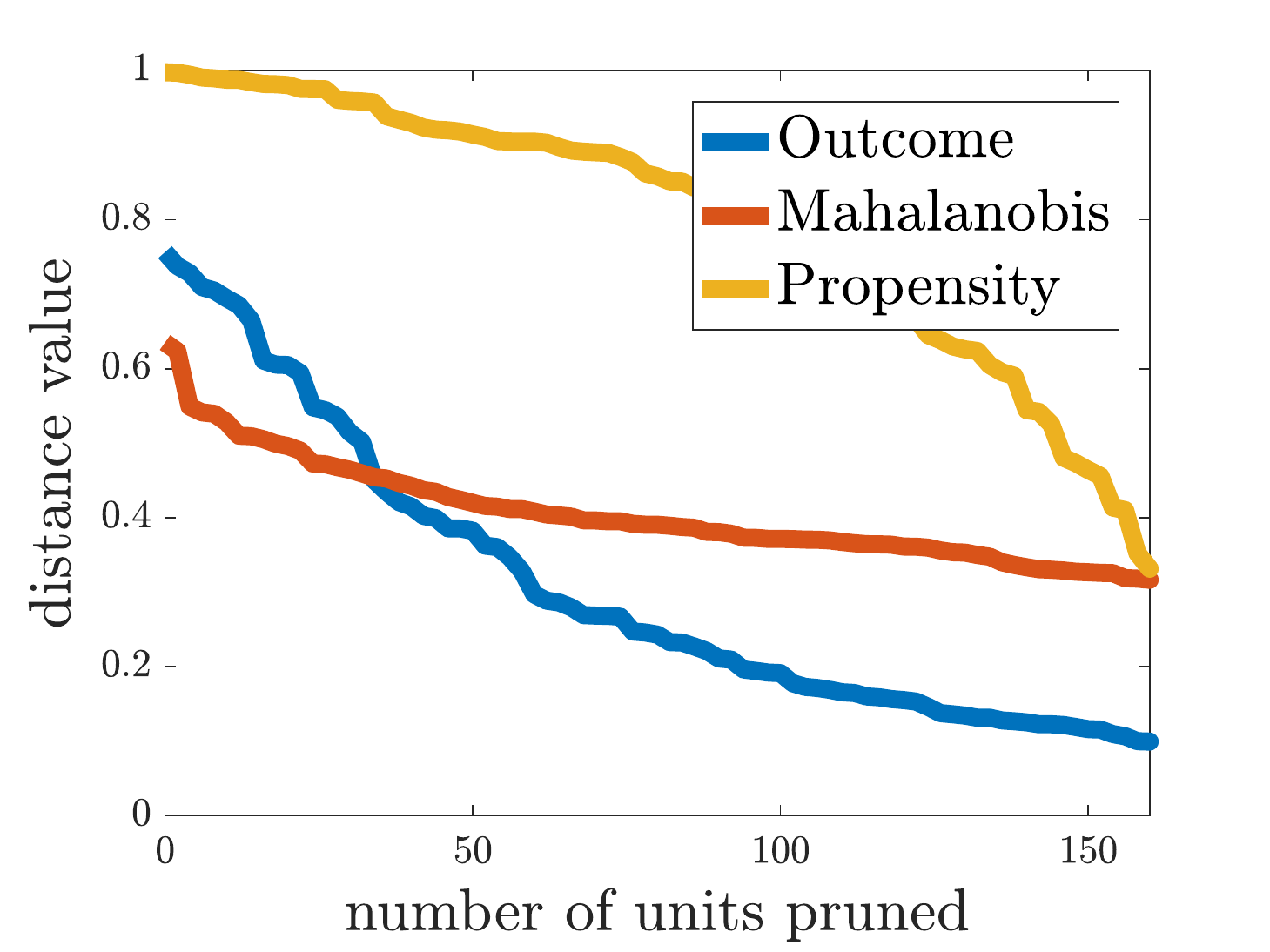} &
			\includegraphics[width=0.3\linewidth]{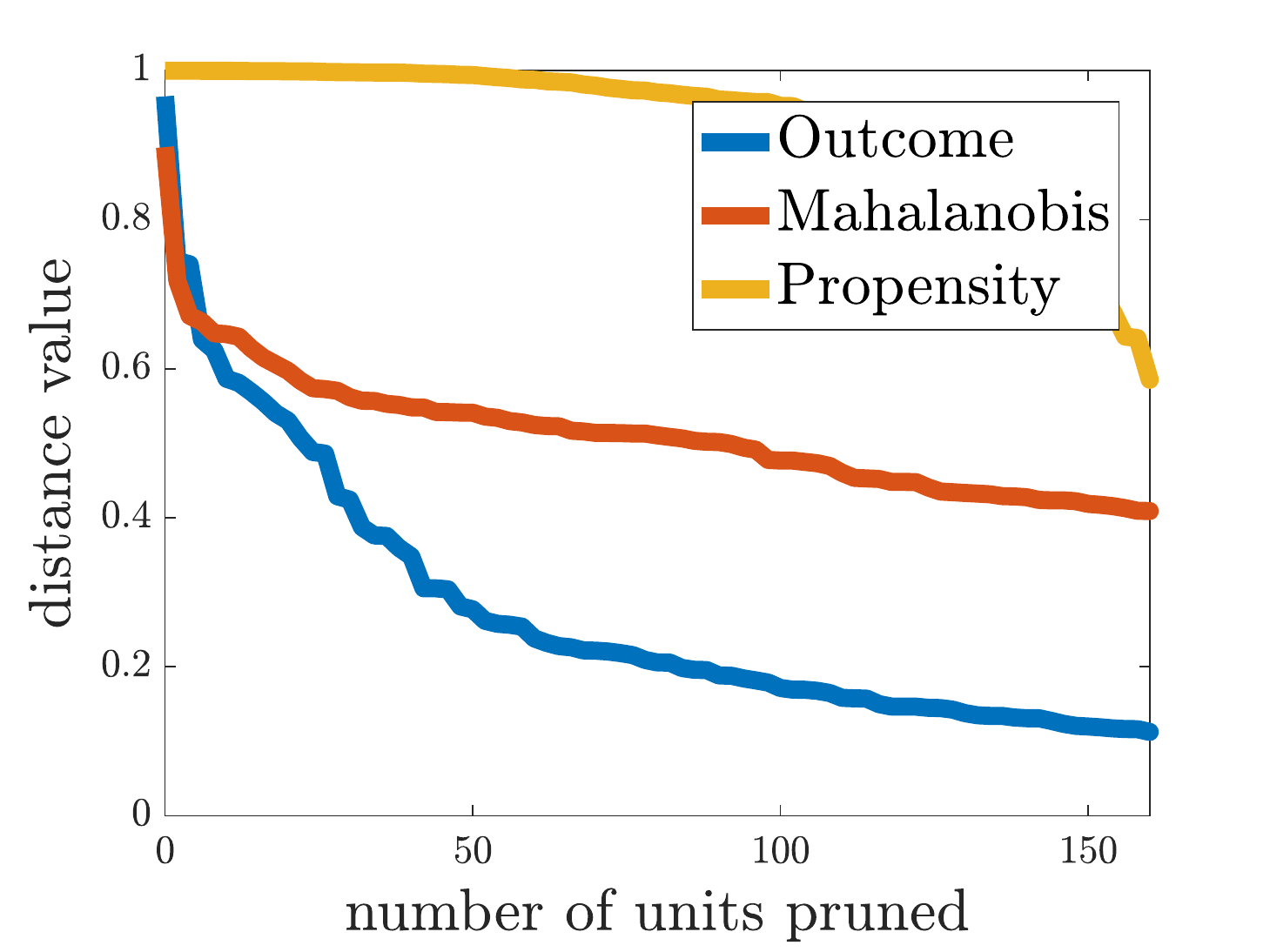}
		\\
		Scenario 7 & Scenario 8 & Scenario 9
	\end{tabular}
	\caption{Distance values over the number 
		of units pruned for scenarios 7-9.}
	\label{fig:sc_789_dc}
\end{figure}
\normalsize

%
%


\section{Discussion and Conclusion}

This paper gives a simple and concrete algorithm for generating matches that discounts or removes covariates that are irrelevant to the outcomes and propensity of a person, and performs a weighted distance matching on the remaining covariates.  The matching is balanced, and empirically has higher accuracy and lower variance than propensity matching and Mahalanobis distance matching for a wide range of calipers, especially in settings where there exist covariates irrelevant to the outcome and propensity functions.  This builds upon the work of King and Nielson \cite{kingnielson} in further demonstrating that propensity matching is highly sensitive to the caliper, and to non-constant treatment effects.

This method lends itself to a number of extensions toward personalized treatment predictions, mostly because the method groups together people that would have an outcome at the same rate pre-treatment.  The authors are examining personalized treatment recommendations based off of the resulting OBM matching, especially when the treatment effect functional form is unknown, using either function driven diffusion metrics for counterfactual functions \cite{cloninger2016function} or deep neural network predictions for survival data \cite{katzman2016deep}.  The authors are also exploring lower bounds for the algorithm and more complex models of patient covariates.  The extension of the algorithm to binary outcomes and local averages for prediction is also a subject of future work.





%
%
%
%
%
%
%

\newpage

\bibliographystyle{wileyj}
\bibliography{main}

\newpage

 \appendix
 
 \section{Regular Treatment Assignment}
 \label{s:reg_assign}
 
 For completeness, we define the three assertions of a 
 \emph{regular} treatment assignment (cf.\,\cite{imbensrubinbook}).
 With notation  
 $\mathbf{X} = (X_1,\dotsc,X_N)$, 
 $\mathbf{T} = (T_1,\dotsc,T_N)$,
 $\mathbf{Y}(1) = (Y_1(1),\dotsc,Y_N(1))$,
 $\mathbf{Y}(0) = (Y_1(0),\dotsc,Y_N(0))$, 
 the unit level assignment probability for unit $i$ is defined by
 \begin{equation}
 p_i(\bX,\bY(1),\bY(0)) = \sum_{\bT :\, T_i=1} P(\bT \given \bX, \bY(1), \bY(0)) \,.
 \end{equation}
 \emph{Individualistic} assignment asserts that 
 (i) the probability of assignment to treatment for unit $i$ is some common function $q$
 of unit $i$'s covariates and potential outcomes, that is,
 \begin{equation}
 p_i(\bX,\bY(1),\bY(0)) = q(X_i,Y_i(1),Y_i(0))
 \end{equation}
 for all $i=1,\dotsc,N$,
 and (ii)
 \begin{equation}
 P(\bT \given \bX,\bY(1),\bY(0)) 
 = C \prod_{i=1}^N q(X_i,Y_i(1),Y_i(0))^{T_i}[1-q(X_i,Y_i(1),Y_i(0))]^{1-T_i}
 \end{equation}
 for $(\bX,\bT,\bY(1),\bY(0)) \in \mathbb{A}$, for some set $\mathbb{A}$, and is zero elsewhere.
 \emph{Probabilistic} assignment asserts that $0 < p_i(\bX,\bY(1),\bY(0)) < 1$ everywhere, 
 for all $i = 1,\dotsc,N$.
 \emph{Unconfounded} assignment asserts that 
 $P(\bT \given \bX, \bY(1), \bY(0)) = P(\bT \given \bX, \bY'(1), \bY'(0))$
 for all $\bX, \bY(1),\bY(0),\bY'(1),\bY'(0)$.
 
 Together, these assertions imply the joint probability
 \begin{equation}
 P(\bX,\bT,\bY(1),\bY(0)) = P(\bX,\bY(1),\bY(0)) \prod_{i=1}^N \pi(X_i)^{T_i}[1-\pi(X_i)]^{1-T_i} \,.
 \end{equation}
 

 %
 %
 %
 %
 
 \section{Proof of Lemma \ref{lemma:normalRegression}}\label{appendix:proofLemma}

 Note that $T_i$ is determined as
 \begin{equation} 
 T_i = \begin{cases} 
 1 :& \alpha_0 + X_i \alpha + w > 0 \\
 0 :& \text{otherwise}
 \end{cases},
 \end{equation}
 where $w\sim \mathrm{Logistic}(0,1)$.
 
 From \cite{Tallis}, we have
 \begin{eqnarray*}
 \E[X^T \given X c > p] &=& \left(\Phi(p/\gamma) \gamma\right)^{-1} \phi(p/\gamma) \Sigma c,\\
 \E[XX^T \given Xc > p] &=& \Sigma + \Sigma c c^T \Sigma \left(\Phi(p/\gamma) \gamma^2\right)^{-1}\phi(p/\gamma)  \Big(p/\gamma - \phi(p/\gamma)/\Phi(p/\gamma) \Big),
 \end{eqnarray*}
 where $\gamma = (c^* \Sigma c)^{1/2}$, 
 $\Phi$ is the cdf of a 1D normal random variable, 
 and $\phi$ is the pdf of a 1D normal random variable.  For simplicity, we replace these constants as
  \begin{align}
 	\E[X^T \given X c > p] &= C_{\Sigma,c,p}\cdot \Sigma c, \label{eq:truncExp}\\
 	\E[XX^T \given Xc > p] &= \Sigma + C'_{\Sigma,c,p} \cdot \Sigma c c^T \Sigma \label{eq:truncCov}
 \end{align}

With these in mind, we consider the expectation over $(\bX^T \bX)^{-1} \bX^T \bT$, 
%
\begin{align}
\E [ (\bX^T \bX)^{-1} \bX^T \bT ] 
&= \E_{X \mid T=1} [ (X^T X)^{-1} X^T \given T=1 ] \\
&= \E_w \E_{X \mid \alpha_0 + X \alpha + w > 0} [ (X^T X)^{-1} X^T \given \alpha_0 + X \alpha + w > 0] \\
&\stackrel{p}{\rightarrow} \E_w \Big(\E_{T=1}[X^T X]\Big)^{-1} \E_{T=1}[X^T],\label{eq:truncProd}
\end{align}
where the last convergence step comes from the continuous mapping theorem.  This means we can treat each term separately and consider the product at the end.  By \eqref{eq:truncExp}, the right term gives us
\begin{align}
\E_{T=1}[X^T] 
&= \E[X^T \given  X \alpha > -\alpha_0 - w]\\
&= C_{\Sigma,\alpha,-\alpha_0-w} \cdot \Sigma \alpha,
\end{align}
and the left term gives us
\begin{align}\label{eq:truncCovAlpha}
\E_{T=1}[X^T X] 
&= \E[X^T X \given  X \alpha > -\alpha_0 - w]\\
&= \Sigma + C'_{\Sigma,\alpha,-\alpha_0-w} \cdot \Sigma \alpha \alpha^T \Sigma.
\end{align}
Now we consider the inversion of \eqref{eq:truncCovAlpha}, which gives us
\begin{align}
\Big(\E_{T=1}[X^T X]\Big)^{-1} &= \Sigma^{-1} -  c\Big(I + c\cdot \alpha \alpha^T \Sigma\Big)^{-1} \alpha \alpha^T,
\end{align}
for constant $c$ depending on $\Sigma,\alpha,\alpha_0,w$, by the Binomial inverse theorem (a generalization of the matrix inversion lemma because $\alpha\alpha^T$ isn't invertible).  

We now consider the product in \eqref{eq:truncProd},
\begin{align}
\Big(\E_{T=1}[X^T X]\Big)^{-1} \E_{T=1}[X^T] 
&= \bigg[\Sigma^{-1} -  c\Big(I + c\cdot \alpha \alpha^T \Sigma\Big)^{-1} \alpha \alpha^T\bigg] C_{\Sigma,\alpha,-\alpha_0-w}\cdot  \Sigma \alpha\\
&= C_1 \alpha - C_2 \Big(I + c\cdot \alpha \alpha^T \Sigma\Big)^{-1} \alpha \cdot (\alpha^T \Sigma \alpha) \\
&= C_1 \alpha - C_3 \Big(I + c\cdot \alpha \alpha^T \Sigma\Big)^{-1} \alpha,
\end{align}
where $C_3$ now depends on $\alpha$.

Now we assume the support of $\alpha$ is sparse, so we can write without loss of generality 
\begin{align}
\alpha &= \begin{bmatrix} A \\ 0 \end{bmatrix} \\
I + c\cdot \alpha \alpha^T \Sigma &= \begin{bmatrix}
I + c AA^T \Sigma_{11} & c AA^T\Sigma_{12} \\
0 & I
\end{bmatrix},
\end{align}
where $\Sigma = \begin{bmatrix} \Sigma_{11} & \Sigma_{12}\\\Sigma_{21} &\Sigma_{22}\end{bmatrix}$.  We know from block $2\times 2$ matrix inversion, 
\begin{align}
\begin{bmatrix}
I + c AA^T \Sigma_{11} & c AA^T\Sigma_{12} \\
0 & I
\end{bmatrix}^{-1} &= \begin{bmatrix}
(I + c AA^T \Sigma_{11})^{-1} & D\\0 & I
\end{bmatrix},
\end{align}
for $D$ written as a function of $I,A,\Sigma$.  Collecting all terms, this means
\begin{align}
\E [ (\bX^T \bX)^{-1} \bX^T \bT ] 
&\stackrel{p}{\rightarrow}  E_w\bigg( C_1 \begin{bmatrix} A\\0\end{bmatrix} - C_3 \begin{bmatrix}
(I + c AA^T \Sigma_{11})^{-1} A\\0
\end{bmatrix} \bigg) \\
&=\begin{bmatrix}
c\cdot  A - c'\cdot (I + c AA^T \Sigma_{11})^{-1} A\\0
\end{bmatrix}.
\end{align}

Thus completes the proof that, if $\alpha_i=0$, then the same is true for the $i^{th}$ element of the treatment regression coefficients.


 

 \section{Proof of Theorem \ref{thm:noiseFolding}}
 
 
 We begin by reforming the expected value of the Mahalanobis distance between a points $x$ and $y$ where $T_x = 1$ and $T_y = 0$.  This yields
 \begin{eqnarray*}
 	\E\left[ d_R(x, y)^2\right] &=&\E \left[ (x - y)^* \Sigma^{\dagger} (x - y) \right]\\
 	&=& \E \left[\sum_{j,k} (\Sigma^{\dagger})_{j,k} (x_j - y_j) (x_k - y_k) \right] \\
 	&=&\sum_{j,k} (\Sigma^{\dagger})_{j,k} \E \left[(x_j - y_j)(x_k - y_k)\right].
 \end{eqnarray*}
 
 With $x_j \sim U([\eta_j,\eta_j+1])$, $y_j \sim U([0,1])$, and $\xi_j := x_j - y_j$,
 we have
 $\E[\xi_j] = \eta_j$,
 $\E[\xi_j \xi_k] = \E[\xi_j] \E[\xi_k] = \eta_j \eta_k$,
 and $\E[\xi_j^2] = \frac{1}{6} + \eta_j^2$ by integration.
 
 Now consider two points $X_i$ and $X_{m(i)}$ that are a perfect matching, so $(X_i)_k = (X_{m(i)})_k$ for $k\in K$.  Without loss of generality, assume $(X_i)_k = 0$ and assume $K = \{1,..., d\}$.  Then
 \begin{eqnarray*}
 	\E\left[d_R(X_i, X_{m(i)})^2\right] &=& \sum_{j,k} (\Sigma^{\dagger})_{j,k} \E \left[((X_i)_j - (X_{m(i)})_j)((X_i)_k - (X_{m(i)})_k)\right]\\
 	&=& \sum_{j,k \not\in K} (\Sigma^{\dagger})_{j,k} \E \left[((X_i)_j - (X_{m(i)})_j)((X_i)_k - (X_{m(i)})_k)\right] \\
 	&=& \sum_{j \not\in K} (\Sigma^{\dagger})_{j,j}  \E ((X_i)_j - (X_{m(i)})_j)^2 + ...\\
 	&&\sum_{j,k \not\in K, j\neq k} (\Sigma^{\dagger})_{j,k} \E \left[((X_i)_j - (X_{m(i)})_j)((X_i)_k - (X_{m(i)})_k)\right] \\
 	&=& \sum_{j\not\in K} \frac{1}{6} (\Sigma^{\dagger})_{j,j} + 
 	\sum_{j,k\not\in K} (\Sigma^{\dagger})_{j,k} \eta_j \eta_k.
 \end{eqnarray*}

\end{document}